\documentclass[numbers,traditional]{ima-authoring-template}%

\theoremstyle{thmstyletwo}%
\newtheorem{theorem}{Theorem}
\newtheorem{proposition}[theorem]{Proposition}%
\newtheorem{lemma}[theorem]{Lemma}%
\newtheorem{corollary}[theorem]{Corollary}%
\newtheorem{example}{Example}%
\newtheorem{remark}{Remark}%
\newtheorem{definition}{Definition}%

\numberwithin{equation}{section}

\DeclareMathOperator{\Tr}{Tr}\DeclareMathOperator{\rank}{rank}
\DeclareMathOperator{\Ran}{Ran}\DeclareMathOperator{\diag}{diag}
\newcommand{\Om}{\Omega}\newcommand{\Vq}{V_{\mathrm q}}\newcommand{\half}{\tfrac12}
\newcommand{\Vc}{V_{\mathrm c}}\newcommand{\Nq}{N_{\mathrm q}}
\newcommand{\Jc}{J}\newcommand{\R}{\mathbb R}\newcommand{\C}{\mathbb C}
\newcommand{\Sp}{\mathrm{Sp}}\newcommand{\Sym}{\mathrm{Sym}}

\makeatletter 
\g@addto@macro\ps@opening{\def\@oddhead{}
\def\@evenhead{}} 
\makeatother  
\begin{document}

\journaltitle{} 
\firstpage{1}


\title[Quantum complexity resource in GBS]{Quantum complexity resource in Gaussian boson sampling: Core structure of the semidefinite program}

\author{Kunwar Kalra
\address{\orgdiv{Department of Mathematics}, \orgname{Texas A\&M University}, \orgaddress{\street{College Station}, \postcode{77843}, \state{Texas}, \country{USA}}}}
\author{V.~V.~Kocharovsky
\address{\orgdiv{Department of Physics \& Astronomy}, \orgname{Texas A\&M University}, \orgaddress{\street{College Station}, \postcode{77843}, \state{Texas}, \country{USA}}}}

\authormark{Kalra and Kocharovsky}


\abstract{We present a rigorous analysis of the algebraic and geometric structure of the quantum complexity resource of a system of bosonic modes in Gaussian boson sampling. This resource underlies the quantum advantage of the system: its photon-counting statistics require the evaluation of a hafnian of the resource covariance matrix, and that computation is $\sharp$P-hard. The resource covariance matrix is the solution of a semidefinite program that extracts the minimum-trace physical quantum part of the total covariance matrix; the complementary part is positive semidefinite and can therefore be simulated classically. Earlier work characterized this resource only through the trace of the quantum part, equal to its photon number. We characterize the optimizer itself, as a quantum state and as a geometric object, beyond the scalar given by its trace. We prove that it is a unique pure Gaussian state and construct an explicit oracle map, obeying an algebraic Riccati identity, that reconstructs the resource. We prove that the full problem compresses exactly onto the active symplectic sector that the dual program support generates. The passive-diagonalizable states are solved in closed form, the first explicit solvable class, and the whole program is shown to be equivalent to a minimization over the symplectic group, that is, over the Siegel upper half-space. Together these results establish that the program determines a canonical localized pure Gaussian component of the resource, and they provide the structural foundation for its detailed analysis.} 

\maketitle 
\thispagestyle{empty}

\onecolumn

\section{Introduction}\label{sec:intro}

Revealing a quantum complexity resource responsible for quantum advantage of continuous-variable (CV) quantum systems over classical systems is a central problem in quantum information science \cite{Oh2024,Serafini2023,Weedbrook2012,Zhong2020,Takeda2019,Preskill2018,Boixo2018,RadNature2025}. Many CV quantum systems employ multimode squeezed and entangled light generated via parametric down-conversion or four-wave mixing. The most advanced sources of squeezed entangled multimode light, widely used in quantum optics science and technology, are based on optical parametric amplifiers (OPA), oscillators (OPO), nonlinear waveguides, and interferometers. They usually provide light in a mixed Gaussian state \cite{Serafini2023}. 

Such multimode Gaussian light has quantum statistics that are $\sharp$P-hard to compute, and provides a quantum complexity resource suitable for a full-scale implementation of quantum advantage. A recent example is a series of experiments aimed at demonstrating quantum advantage in Gaussian boson sampling (GBS) \cite{Pan2026,Hamilton2017,Quesada2018,Zhong2019,Huh2019,Wang2019,Brod2019,PanPRL2021,Deshpande2022,Bulmer2022,Madsen2022,Deng2023,Pan2023,Yu2023}. Moreover, the GBS-like sources of multimode Gaussian light generate the cluster states that serve as the initial resource of the first experimental prototype of a one-way CV measurement-based quantum computer, Aurora \cite{RadNature2025}. It employs 84 OPO squeezers, supplying 42 GBS cells with single-mode squeezed light and furnishing 12 physical qubit modes at each clock cycle, and incorporates all major building blocks required for the universal fault-tolerant photonic quantum computer \cite{Takeda2019}. Numerous other setups and applications in modern quantum optics technologies, for instance, in networking and quantum cryptography \cite{Weedbrook2012}, are also based on multimode Gaussian light. 

Thus, it is necessary to disclose a quantum complexity resource hidden in a Gaussian state of $M$ optical modes described by a real symmetric $2M\times2M$ quadrature covariance matrix $V$ (up to a displacement we may set aside):
\begin{equation} \label{VG}
V = \left[ \begin{matrix}
\langle \hat{p}_k \hat{p}_{k'} \rangle 
            &   
\frac{1}{2}\langle \hat{q}_k \hat{p}_{k'}+\hat{p}_{k'} \hat{q}_k \rangle^T
            \\[6pt]
\frac{1}{2}\langle \hat{q}_k \hat{p}_{k'}+\hat{p}_{k'} \hat{q}_k \rangle 
            &   
\langle \hat{q}_k \hat{q}_{k'} \rangle
        \end{matrix} \right] 
= \frac{i}{2}\Omega + \langle \hat{s}\hat{s}^T \rangle. 
\end{equation}
This matrix collects the variances and correlations of the quadrature operators of the modes, listed throughout in the order $\hat{s}= (\hat{p}_1,\dots,\hat{p}_M,\hat{q}_1,\dots,\hat{q}_M)^T$. These operators do not commute. They obey canonical commutation relations $ [\hat{q}_k,\hat{p}_{k'}]=i\delta_{kk'}$ recorded by the \emph{symplectic form} 
\begin{equation} \label{Omega}
\Om=\begin{pmatrix}0&I_M\\-I_M&0\end{pmatrix},\qquad \Om^\top=-\Om,\qquad \Om^2=-I_{2M}.
\end{equation}
Heisenberg's uncertainty principle, in its multimode (Robertson--Schr\"odinger) matrix form, states that a real symmetric $V$ is the covariance matrix of an actual quantum state, in which case we call it \emph{physical}, exactly when the Hermitian matrix $V+\tfrac{i}{2}\Om$ is positive semidefinite, $V+\tfrac{i}{2}\Om\succeq0$. The least uncertain state, the vacuum (zero photons), corresponds to $V=\half I_{2M}$.

An important step toward revealing the quantum complexity resource was taken recently by Oh and co-authors \cite{Oh2024}, who introduced a classical algorithm that efficiently computes photon-counting statistics of a noisy Gaussian state if the number of squeezed photons, whose statistics require the evaluation of a hafnian that is $\sharp$P-hard to compute, is below the largest hafnian size that classical computers can still evaluate, on the order of one hundred. 
This identification of quantum advantage, or quantum supremacy, with the $\sharp$P-hardness of computing the $\sharp$P-complete hafnian is justified by the hafnian master theorem \cite{LAA2022,PRA2022} and Toda's theorem on the $\sharp$P-complete oracle \cite{Toda1991,Basu2012} and goes far beyond the GBS problem \cite{Entropy2026,OPAarXiv2026}. 

As a result, Oh and co-authors \cite{Oh2024} obtain a scalar expression for the quantum complexity resource, identifying it with the number of such incomputable squeezed photons, $\Nq(V)=\half\bigl(\Tr\Vq^\star-M\bigr)$, given by the trace of the covariance matrix's quantum part $\Vq^\star$ which minimizes this trace subject to two constraints: 
\begin{equation}\label{eq:sdp_full}
\Vq^\star=\arg\min\bigl\{\,\Tr(\Vq)\ :\ \Vq+\tfrac{i}{2}\Om\succeq0,\ \ V-\Vq\succeq0,\
\Vq\in\Sym_{2M}(\R)\,\bigr\}.
\end{equation}

In other words, the classical algorithm of Oh and co-authors\ \cite{Oh2024} isolates the nonclassical content of $V$, responsible for the $\sharp$P-hard part of the task, by means of the semidefinite program (SDP) \cite{BenTalNemirovski}, that is, a convex optimization over positive-semidefinite matrices, (\ref{eq:sdp_full}) that splits $V$ into the smallest physical \emph{quantum} part $\Vq$ and an admissible \emph{classical} noise remainder $\Vc$. The first constraint requires the quantum part $\Vq$ to be physical. The second requires the remainder $\Vc=V-\Vq$ to be a positive matrix, which is precisely the condition for it to be admissible classical noise (a random classical displacement, which is a free operation). Among all such decompositions the program selects the one of least trace, assigning all irreducible nonclassicality to $\Vq^\star$ and the remainder to classical noise. The associated quantum photon number $\Nq(V)$ counts photons above the vacuum and sets a certified lower bound on the classical cost of sampling the state. 

Previously, the quantum complexity resource of the multimode Gaussian light was evaluated on the basis of the total mean number of photons in the Bloch--Messiah eigen-squeezed modes (supermodes) associated with the vacuum of Bogoliubov quasiparticles. Referring to the quasiparticle vacuum is natural within canonical quantum statistical physics and quantum field theory, but does not account for the quantum-information aspect of the problem, in particular the computational $\sharp$P-hardness of quantum statistics of multimode systems. The paper \cite{Oh2024} shows that the common approach to estimating quantum complexity, based on the Bloch--Messiah supermodes, is misleading, overstates quantum complexity, and is responsible for false claims of quantum advantage in the recent GBS experiments \cite{PanPRL2021,Madsen2022,Deng2023}. 

Minimizing a function subject to equality and inequality constraints, as with $\Tr(\Vq)$ in Eq.~(\ref{eq:sdp_full}), is carried out by the method of Lagrange multipliers, which underlies analytical mechanics, quantum statistical physics, and field theory. 
It allows one to incorporate all of the constraints into a single Lagrangian function instead of treating them separately. 
Equating to zero the first derivatives of the Lagrangian function with respect to both the original variables and the Lagrange multipliers yields the system of equations for the Lagrangian's saddle point which determines the extremum. 
A standard application of this method is the derivation of the grand canonical ensemble from the canonical ensemble in the quantum statistical physics of Bose--Einstein condensation \cite{Zubarev1974,JStatPhys2015}. There the constraint is the conservation of the total number of particles, $\hat{N} = \mathrm{const}$, the chemical potential $\mu$ plays the role of the Lagrange multiplier, and their product $\mu \hat{N}$ constitutes an additional term in the Lagrangian function.  

Accounting for the inequality constraints entering the convex optimization (\ref{eq:sdp_full}) is done via the Karush–Kuhn–Tucker (KKT) theorem which generalizes the original Lagrange method of multipliers to the case of inequality constraints \cite{BenTalNemirovski}. Because the cost function $\Tr(\Vq)$ depends on the real symmetric matrix $\Vq$ instead of a scalar variable, the Lagrange multiplier must also be a matrix $S$ and to enter the Lagrangian function via the inner product, $\Tr(S(V-\Vq))$,  with the matrix $V-\Vq$ constituting the left side of the inequality constraint $V-\Vq\succeq0$ in Eq.~(\ref{eq:sdp_full}). Finally, the SDP in Eq.~(\ref{eq:sdp_full}) resembles a real-valued SDP because the real symmetric matrices with a positive semidefinite difference obey the L\"owner order.  

Earlier works used the program (\ref{eq:sdp_full}) almost entirely through the single number $\Nq(V)$. The present paper analyzes the optimizer $\Vq^\star$ itself, as a quantum state and as a geometric object determined by $V$, beyond the scalar given by its trace. We show that the program determines, canonically and uniquely, a \emph{pure} Gaussian state, a minimum-uncertainty component carrying no residual classical mixing, reconstructible from the dual data by an explicit algebraic formula. The classical remainder is forced to be singular along a definite number of directions, and once the relevant dual support is fixed the entire $2M$-dimensional problem reduces \emph{exactly} to a low-dimensional symplectic sector that contains all of the resource squeezing. (Photons and modes involved in the commonly examined Bloch--Messiah supermode squeezing of the quasiparticle vacuum could occupy the entire covariance $V$, or a much larger part of it, and could be significantly larger in number than those involved in the resource squeezing.)

Concretely, we establish the following structural facts. 

\emph{(i)} Every optimizer is pure (Theorem~\ref{thm:purity}). 

\emph{(ii)} The optimizer is unique (Theorem~\ref{thm:uniqueness}). 

\emph{(iii)} The inner minimization has a closed-form solution, the oracle $\Vq(A)$ (Theorem~\ref{thm:oracle}). 

\emph{(iv)} That oracle obeys an algebraic Riccati identity (Theorem~\ref{thm:riccati}), and a dual optimum reconstructs the primal one through it (Corollary~\ref{cor:dual_to_primal_reconstruction}). 

\emph{(v)} The kernel of the classical remainder has multiplicity at least $\kappa(V)$, the number of sub-vacuum directions of the covariance $V$ (Theorem~\ref{thm:kernel}). 

\emph{(vi)} The passive-diagonalizable states are solved in closed form (Theorem~\ref{thm:passive_soln}). 

\emph{(vii)} A fixed dual support reduces the problem exactly to its active symplectic sector (Theorem~\ref{thm:active_reduction}). 

\emph{(viii)} The whole program is equivalent to an optimization over $\Sp(2M)/\mathrm U(M)$ (Theorem~\ref{thm:symp_reform}).

Throughout, every structural identity and closed form is checked against a direct numerical solution of the primal SDP, and the tolerances quoted come from those checks.

Several directions build on this structural foundation and go beyond the first solvable class presented here: a fast solver based on the exact active reduction developed below; and the two-mode coupled sector, where the closed forms presented below no longer apply and a Galois-theoretic obstruction arises. These are left to future work.

The presentation proceeds in the following order: duality first, then purity and uniqueness, then the oracle and Riccati structure, then the kernel theorem, the passive closed form, and finally the exact active reduction and the geometric reformulation.

\section{Gaussian covariance preliminaries}\label{sec:prelim}

The natural symmetries of the problem are the \emph{symplectic} linear maps, the changes of quadrature frame (basis) that preserve the commutation relations: the group
$\Sp(2M,\R)=\{S\in\mathrm{GL}(2M,\R):S\Om S^\top=\Om\}$. Williamson's
theorem~\cite{Williamson1936} is the spectral theorem adapted to this group: every positive-definite matrix $A$ can be brought by a symplectic congruence
$SAS^\top=\diag(\nu_1,\dots,\nu_M,\nu_1,\dots,\nu_M)$ to a diagonal form in which each value is repeated once in a $q$-slot and once in the matching $p$-slot. The positive numbers $\nu_1\le\cdots\le\nu_M$ are the \emph{symplectic eigenvalues} of $A$. Equivalently, they are the positive numbers $\nu_j$ for which $\pm\nu_j$ are the eigenvalues of $i\Om A$, and they are the invariants of $A$ under symplectic changes of frame. For a covariance matrix (\ref{VG}) the uncertainty principle is exactly the statement that every symplectic eigenvalue is at least $\half$, and physicality $V+\tfrac{i}{2}\Om\succeq0$ is equivalent to $\nu_j(V)\ge\half$ for all $j$. Modes at the vacuum floor, $\nu_j=\half$, saturate the uncertainty relation. We call a state \emph{pure} when \emph{all} of its symplectic eigenvalues equal $\half$, that is, all Bogoliubov quasiparticles are in the vacuum state with zero average occupation. A passive (photon-number-conserving) frame change is a symplectic map that is also orthogonal; these form the maximal compact subgroup $K(2M):=\mathrm O(2M)\cap\Sp(2M,\R)$, isomorphic to the unitary group  $\mathrm U(M)$.

The program acts only on the directions in which $V$ has less variance than the vacuum. The two spectra of $V$ enter in different roles: the symplectic eigenvalues $\nu_j(V)$, invariant under every symplectic frame change, fix physicality through $\nu_j\ge\half$, while the ordinary eigenvalues $\lambda_j(V)$, invariant only under passive (orthogonal) frame changes, fix where $V$ drops below the vacuum variance. Admissible classical noise can only add covariance, so the program is governed by the ordinary spectrum through the sub-vacuum count $\kappa(V)$ defined next.

\begin{definition}[Sub-vacuum eigenspace]\label{def:subvac}
For physical covariance $V$, let $E_<(V)=\bigoplus_{\lambda_j(V)<1/2}\ker(V-\lambda_jI)$ be the span of the eigenvectors of $V$ whose (ordinary) eigenvalue dips below the vacuum floor, and let
$\kappa(V)=\dim E_<(V)$ be the number of such \emph{sub-vacuum} directions.
\end{definition}

\section{The semidefinite program and its dual}\label{sec:sdp}

The program \eqref{eq:sdp_full} is a constrained minimization, called the \emph{primal}. Associated with it is a second optimization, the \emph{dual}, assembled from the constraints; its variable is a positive-semidefinite matrix $S\succeq0$, a Lagrange multiplier for the inequality $V-\Vq\succeq0$. To state it we first isolate the inner problem.

\begin{definition}[The physical cone and the dual objective]\label{def:gbs_sdp}
Write $\mathcal Q=\{\Vq\in\Sym_{2M}(\R):\Vq+\tfrac{i}{2}\Om\succeq0\}$ for the physical quantum covariances. For a positive-definite weight $A\succ0$ let
$\Phi(A):=\min_{\Vq\in\mathcal Q}\Tr(A\Vq)$ be the least weighted trace of a physical state under that weighting, and define the \emph{dual objective}
\begin{equation}\label{eq:dual}
D(S):=-\Tr(SV)+\Phi(I+S),\qquad S\succeq0 .
\end{equation}
\end{definition}
Eq.~(\ref{eq:dual}) originates from the Lagrangian function, 
$\Tr(\Vq) - \Tr(S(V-\Vq))$, of the KKT theorem for the program \eqref{eq:sdp_full}.
The dual is the maximization $\max_{S\succeq0}D(S)$, and \emph{weak duality} is the elementary fact that every value of the dual is a lower bound for every value of the primal. A \emph{dual certificate} is a choice of $S$ whose dual value $D(S)$ equals the trace of a
candidate primal solution; since every dual value lower-bounds every primal value, such an $S$ forces the unknown optimum to equal that common value and thereby \emph{proves} optimality with no further search and without reliance on a numerical solver.

\begin{definition}[Slater's condition \cite{BenTalNemirovski}]\label{def:slater}
We say the covariance $V$ is \emph{strictly mixed} if all of its symplectic eigenvalues exceed $\half$. This is Slater's condition for \eqref{eq:sdp_full}, the standard requirement that the feasible set have an interior point: some quantum state lies strictly below $V$ in the L\"owner order.
\end{definition}

\begin{theorem}[Reduced duality and the optimality conditions]\label{thm:dual_kkt}
For physical covariance $V$ and the dual objective \eqref{eq:dual}, the following hold.
\begin{enumerate}
\item\emph{(Weak duality.)} $D(S)\le\min\{\Tr\Vq:\Vq\in\mathcal Q,\ \Vq\preceq V\}$ for every $S\succeq0$, so the reduced dual is $\max_{S\succeq0}D(S)$.
\item\emph{(Strong duality and attainment.)} If $V$ is strictly mixed, the primal and dual optima coincide and a maximizer $S^\star\succeq0$ exists.
\item\emph{(Optimality conditions.)} Under the same condition, a pair $(\Vq^\star,S^\star)$ is primal--dual optimal if and only if
\begin{equation}\label{eq:kkt_reduced}
S^\star\succeq0,\qquad \Vq^\star=\Vq(I+S^\star),\qquad V-\Vq^\star\succeq0,\qquad
\Tr\!\bigl(S^\star(V-\Vq^\star)\bigr)=0,
\end{equation}
where $\Vq(\cdot)$ is the oracle of Theorem~\ref{thm:oracle}. In particular $\Ran(S^\star)\subseteq\ker(V-\Vq^\star)$.
\end{enumerate}
\end{theorem}

The four conditions in \eqref{eq:kkt_reduced} are the concrete certificate. The first states that the multiplier is admissible; the second states that $\Vq^\star$ is the oracle output for the weight $I+S^\star$; the third is primal feasibility; and the last is \emph{complementary slackness}, the requirement that the multiplier $S^\star$ act only where the constraint $V-\Vq^\star\succeq0$ is tight, that is, on the kernel of $V-\Vq^\star$, the directions along which no classical noise remains. The boundary case in which $V$ is itself already pure, where the feasible set degenerates to the single point $\Vq^\star=V$, is recovered by continuity.

\begin{proof}
For $S\succeq0$ the Lagrangian of the primal is
$L(\Vq,S)=\Tr(\Vq)+\Tr(S(\Vq-V))=-\Tr(SV)+\Tr((I+S)\Vq)$, and taking the infimum over $\Vq\in\mathcal Q$ gives exactly $D(S)=-\Tr(SV)+\Phi(I+S)$, a lower bound on the primal value; this is weak duality. Under Slater's condition the realified problem (Lemma~\ref{lem:realify}) is a strictly feasible finite-dimensional real SDP, and standard convex duality (Lemma~\ref{lem:strong_duality}) supplies strong duality and a dual maximizer $S^\star\succeq0$. For the weight $A=I+S^\star\succ0$ the inner minimization defining $\Phi(A)$ has the \emph{unique}
minimizer $\Vq(A)$ by Theorem~\ref{thm:oracle}, so optimality is equivalent to $\Vq^\star=\Vq(I+S^\star)$ together with feasibility $V-\Vq^\star\succeq0$ and complementary slackness $\Tr(S^\star(V-\Vq^\star))=0$. Finally, two positive matrices with zero trace product annihilate each other (Lemma~\ref{lem:kernel_cs}), so the slackness condition forces $\Ran(S^\star)\subseteq\ker(V-\Vq^\star)$.
\end{proof}

The weak duality, strong duality, and optimality content here does not require the closed form of the inner problem; the reconstruction $\Vq^\star=\Vq(I+S^\star)$ is the only place the oracle enters, and it
does so as a consequence of the oracle theorem proved below rather than as an independent input.

\section{Purity and uniqueness of the optimizer}\label{sec:purity}

The feasible set of \eqref{eq:sdp_full} contains every physical state below $V$ -- thermal, squeezed-thermal, displaced. The first structural fact is that the minimizer is always a minimum-uncertainty pure
state. Purity is a consequence of optimality: any thermal excess can be removed while preserving feasibility and strictly lowering the trace.

\begin{theorem}[Purity]\label{thm:purity}
If $V$ is physical and \eqref{eq:sdp_full} is feasible, then every optimizer $\Vq$ is pure: all of its symplectic eigenvalues equal $\half$, equivalently
\begin{equation}\label{eq:purity_id}
\Vq\,\Om\,\Vq=\tfrac14\,\Om .
\end{equation}
\end{theorem}

\begin{proof}
Let $\Vq$ be any feasible quantum part, and bring it to Williamson normal form $\Vq=S^\top DS$ with $S\in\Sp(2M)$ and $D=\diag(\nu_1,\dots,\nu_M,\nu_1,\dots,\nu_M)$ its symplectic spectrum, every $\nu_j\ge\half$ by physicality. Reduce every symplectic eigenvalue to the vacuum floor: replace $D$ by $D'=\half I_{2M}$ and set $\Vq'=S^\top D'S=\half S^\top S$. Since $D'\le D$ entry-wise, congruence by $S^\top$ gives $\Vq'\preceq\Vq$, hence $V-\Vq'\succeq V-\Vq\succeq0$, so the reduced state $\Vq'$ is still feasible, and all of its symplectic eigenvalues equal $\half$. Now $\Vq-\Vq'=S^\top(D-D')S\succeq0$, and this matrix is nonzero the moment some $\nu_j>\half$, in which case $\Tr(\Vq')<\Tr(\Vq)$ strictly. A feasible point of smaller trace cannot be optimal, so
at the optimum every symplectic eigenvalue equals $\half$; the optimizer is pure. Finally, from $\Vq=\half S^\top S$ and the symplectic identity $S^\top\Om S=\Om$ one reads off
$\Vq\Om\Vq=\tfrac14 S^\top(S\Om S^\top)S=\tfrac14 S^\top\Om S=\tfrac14\Om$, which is \eqref{eq:purity_id}.
\end{proof}

\begin{remark}[Interpretation of purity]\label{rem:purity_meaning}
At the optimum the quantum part contains no mixedness: any thermal excess can be removed while strictly decreasing the trace. The trace satisfies $\Tr\Vq=M+2\Nq(\Vq)$ with $\Nq(\Vq)=\tfrac12(\Tr\Vq-M)$ the mean photon number of the quantum part, so minimizing the trace minimizes that photon number and brings every Williamson mode to the vacuum floor; the feasible states admitting no further photon removal under $V-\Vq\succeq0$ are exactly the pure ones. The optimization therefore assigns the entire quantum resource to a single pure Gaussian component, and $\Vq^\star$ is the irreducible nonclassical part of $V$.
\end{remark}

The purity identity \eqref{eq:purity_id} has a geometric interpretation used throughout the rest of the paper.

\begin{corollary}[The induced complex structure]\label{cor:complex}
For a pure covariance $\Vq$, the symplectic matrix $\Jc:=2\Om\Vq$ satisfies $\Jc^2=-I_{2M}$ and $\Jc^\top\Om\Jc=\Om$; it is a \emph{complex structure} (a real linear map squaring to $-I$, which lets one regard the real quadrature space as a complex one), compatible with $\Om$ in the sense that $g_\Jc(x,y):=-x^\top\Om\Jc y=2\,x^\top\Vq y$ is a positive-definite inner product. Operationally, $\Jc$ acts as multiplication by $i$ on phase space: its $\pm i$ eigenspaces are the annihilation and creation subspaces of the modes that diagonalize $\Vq$, so a choice of pure $\Vq$ is equivalent to a choice of which quadrature combinations serve as lowering and raising operators.
\end{corollary}

\begin{proof}
Using \eqref{eq:purity_id}, $\Jc^2=4\Om\Vq\Om\Vq=4\Om(\tfrac14\Om)=\Om^2=-I$. Next
$\Jc^\top\Om\Jc=(-2\Vq\Om)\Om(2\Om\Vq)=4\Vq\Om\Vq=\Om$. And
$g_\Jc(x,y)=-x^\top\Om\Jc y=-2x^\top\Om^2\Vq y=2x^\top\Vq y$, which is positive definite because $\Vq\succ0$.
\end{proof}

Purity can be expressed by identifying the pure covariances with a homogeneous space.

\begin{lemma}[Pure states as a symplectic quotient]\label{lem:pure_param}
A matrix $\Vq\in\Sym_{2M}^{++}$ is pure if and only if $\Vq=\half S^\top S$ for some $S\in\Sp(2M)$. 
The map $S\mapsto\half S^\top S$ descends to a bijection
$\Sp(2M)/\mathrm U(M)\xrightarrow{\ \sim\ }\{\text{pure covariances}\}$.
\end{lemma}

\begin{proof}
If $\Vq=\half S^\top S$ with $S$ symplectic, then $S^{-\top}\Vq S^{-1}=\half I$, so all symplectic eigenvalues equal $\half$ and $\Vq$ is pure. Conversely, if $\Vq$ is pure, Williamson's theorem
gives $T\in\Sp(2M)$ with $T\Vq T^\top=\half I$, whence $\Vq=\half(T^{-1})(T^{-1})^\top$ and
$S:=T^{-\top}$ works. If $\half S_1^\top S_1=\half S_2^\top S_2$, then $K:=S_2S_1^{-1}$ is symplectic with $K^\top K=I$, hence $K\in\Sp(2M)\cap\mathrm O(2M)=\mathrm U(M)$. Thus, the fiber of the map is exactly a left coset of $\mathrm U(M)$, giving the stated bijection.
\end{proof}

For a given Gaussian state there is exactly one way to extract the minimum quantum resource.

\begin{theorem}[Uniqueness]\label{thm:uniqueness}
For every physical $V$, the program \eqref{eq:sdp_full} has a unique optimizer $\Vq^\star$.
\end{theorem}

\begin{proof}
The feasible set $\mathcal F(V)=\{X\in\Sym_{2M}(\R):X+\tfrac{i}{2}\Om\succeq0,\ V-X\succeq0\}$ is nonempty and closed, and every $X\in\mathcal F(V)$ obeys $0\preceq X\preceq V$, so $\mathcal F(V)$
is compact and the trace attains its minimum. Let $\Vq^{(1)},\Vq^{(2)}$ be two optimizers. For
$t\in(0,1)$ the convex combination $\Vq^{(t)}=(1-t)\Vq^{(1)}+t\Vq^{(2)}$ is feasible with the same optimal trace, hence optimal, hence pure by Theorem~\ref{thm:purity}.

Pass to the Hermitian matrices $X_j=\Vq^{(j)}+\tfrac{i}{2}\Om$. Each is positive semidefinite of rank $M$, because a physical state is pure exactly when this matrix has rank $M$
(Lemma~\ref{lem:purity_rank}). For a convex combination of positive matrices the kernel is the intersection of the kernels: if $z$ lies in the kernel of $X_t=(1-t)X_1+tX_2$, then
$0=\langle z,X_tz\rangle=(1-t)\langle z,X_1z\rangle+t\langle z,X_2z\rangle$ forces
$\langle z,X_1z\rangle=\langle z,X_2z\rangle=0$ and hence $z\in\ker X_1\cap\ker X_2$. But $\Vq^{(t)}$ is pure, so $\ker X_t$ has dimension $M$; as $\ker X_1$ and $\ker X_2$ are themselves
$M$-dimensional, their intersection having dimension $M$ forces $\ker X_1=\ker X_2$.

The kernel determines the state. For a pure $\Vq$ with complex structure $\Jc=2\Om\Vq$ (Corollary~\ref{cor:complex}), the kernel of $\Vq+\tfrac{i}{2}\Om$ is the $+i$ eigenspace of $\Jc$; since $\Jc$ is real, its complex conjugate is the $-i$ eigenspace, so the kernel pins down both eigenspaces of $\Jc$ and hence $\Jc$ itself. Applying this to $\Vq^{(1)}$ and $\Vq^{(2)}$, equality of kernels gives $\Jc_1=\Jc_2$, and since $\Vq=-\tfrac12\Om\Jc$ we conclude
$\Vq^{(1)}=\Vq^{(2)}$.
\end{proof}

\begin{remark}[Interpretation of uniqueness]\label{rem:uniqueness_meaning}
Beyond certifying the amount of irreducible nonclassicality, the program selects a single canonical covariance for the quantum component. Any remaining ambiguity is one of coordinates or gauge; the optimizer $\Vq^\star$ itself is determined uniquely.
\end{remark}

\section{The oracle and the Riccati identity}\label{sec:oracle}

Both the duality of Theorem~\ref{thm:dual_kkt} and everything that follows rest on solving the inner minimization $\Phi(A)=\min_{\Vq\in\mathcal Q}\Tr(A\Vq)$ exactly. The solution is a single matrix absolute-value computation. Write $|M|:=\sqrt{M^\top M}$ for the matrix absolute value.

\begin{definition}[Skew transport]\label{def:transport}
For $A\succ0$ set $B(A)=A^{1/2}\Om A^{1/2}$, an invertible skew-symmetric matrix, with absolute value $|B(A)|=\sqrt{-B(A)^2}$.
\end{definition}

We first record a small fact about positive block matrices with rank-one diagonal blocks, used to pin down the off-diagonal blocks in the uniqueness part of the oracle.

\begin{lemma}[Rank-one PSD block factorization]\label{lem:rankone_psd_block}
If $\bigl(\begin{smallmatrix}\alpha uu^\ast&X\\X^\ast&\beta vv^\ast\end{smallmatrix}\bigr)\succeq0$ with $\alpha,\beta>0$ and unit vectors $u,v$, then $X=\sqrt{\alpha\beta}\,c\,uv^\ast$ for some
$c\in\C$ with $|c|\le1$; in particular, if $u=v$ then $X$ is a scalar multiple of $uu^\ast$.
\end{lemma}

\begin{proof}
By the range condition for positive block matrices, $x\perp u$ implies $x\in\ker(\alpha uu^\ast)\subseteq\ker(X^\ast)$, so $\Ran X\subseteq\C u$; testing $\binom{0}{y}$ for $y\perp v$ forces $Xy=0$, so $X=\gamma uv^\ast$. Compressing to
$\operatorname{span}\{\binom{u}{0},\binom{0}{v}\}$ gives
$\bigl(\begin{smallmatrix}\alpha&\gamma\\\bar\gamma&\beta\end{smallmatrix}\bigr)\succeq0$, hence $|\gamma|^2\le\alpha\beta$.
\end{proof}

\begin{theorem}[Oracle]\label{thm:oracle}
For $A\succ0$ the weighted trace $\Tr(A\Vq)$ has a unique minimizer over $\mathcal Q$, namely
\begin{equation}\label{eq:oracle}
\Vq(A)=\half\,A^{-1/2}\,\bigl|B(A)\bigr|\,A^{-1/2},\qquad B(A)=A^{1/2}\Om A^{1/2},
\end{equation}
and the minimum value is $\Phi(A)=\sum_{j=1}^M\nu_j(A)$, the sum of the symplectic eigenvalues of $A$. The minimizer $\Vq(A)$ is pure.
\end{theorem}

We call the map $A\mapsto\Vq(A)$ the \emph{oracle}: for each weighting of the quadratures it returns the pure state of least weighted trace. Once the optimal dual certificate $S^\star$
is known, the quantum part is reconstructed as $\Vq(I+S^\star)$, so the outer optimization reduces to identifying the dual support.

\begin{proof}
Set $R=A^{1/2}$ and change variables to $Y=R\Vq R$; then $\Tr(A\Vq)=\Tr(Y)$, and
$\Vq+\tfrac{i}{2}\Om\succeq0$ becomes $Y+\tfrac{i}{2}B(A)\succeq0$. Since $B(A)$ is skew and invertible, an orthogonal change of basis brings it to canonical $2\times2$ blocks: there is 
$Q\in\mathrm O(2M)$ and $\sigma_1,\dots,\sigma_M>0$ with $Q^\top B(A)Q=\bigoplus_{j=1}^M\sigma_jJ_2$,
where $J_2=\bigl(\begin{smallmatrix}0&1\\-1&0\end{smallmatrix}\bigr)$
(Lemma~\ref{lem:skew_canonical}). Writing $\widetilde Y=Q^\top YQ$, the problem is to minimize $\Tr(\widetilde Y)$ subject to $\widetilde Y+\tfrac{i}{2}\widetilde B\succeq0$, with
$\widetilde B=\bigoplus_j\sigma_jJ_2$.

The block rotation group $G=\mathrm{SO}(2)^M$ commutes with $\widetilde B$, so it preserves the feasible set and the trace; averaging any feasible $\widetilde Y$ over Haar measure on $G$ produces a feasible matrix of the same trace that commutes with all of $G$. The real symmetric matrices commuting with every block rotation are exactly the block-scalar ones
$\overline Y=\bigoplus_j\alpha_jI_2$, so it suffices to minimize over these. The constraint then decouples block-wise to $\alpha_jI_2+\tfrac{i}{2}\sigma_jJ_2\succeq0$, which forces
$2\alpha_j\ge\sigma_j$ with equality iff $\alpha_j=\tfrac{\sigma_j}{2}$ (Lemma~\ref{lem:2x2_bound}). Hence $\Tr(\widetilde Y)=\sum_j2\alpha_j\ge\sum_j\sigma_j$, with the minimizer $\widetilde Y^\star=\bigoplus_j\tfrac{\sigma_j}{2}I_2=\tfrac12|\widetilde B|$.

Uniqueness needs the off-diagonal blocks too. Let $H=\widetilde Y+\tfrac{i}{2}\widetilde B\succeq0$ with $\Tr\widetilde Y=\sum_j\sigma_j$. The blockwise bound forces each diagonal block to be $\widetilde Y_{jj}=\tfrac{\sigma_j}{2}I_2$, so $H_{jj}=\sigma_jP$ with $P=\tfrac12(I_2+iJ_2)=uu^\ast$, $u=\tfrac1{\sqrt2}\binom{1}{-i}$, a rank-one projector. For $j\ne\ell$ the principal $4\times4$ block of $H$ has rank-one diagonal blocks, and a positive matrix
with rank-one diagonal blocks has off-diagonal block a scalar multiple of $uu^\ast$ (Lemma~\ref{lem:rankone_psd_block}); but $H_{j\ell}=\widetilde Y_{j\ell}$ is real while the only real multiple of $uu^\ast=\tfrac12\bigl(\begin{smallmatrix}1&i\\-i&1\end{smallmatrix}\bigr)$ is zero. So $H$ is block diagonal and $\widetilde Y=\tfrac12|\widetilde B|$ is forced. Undoing the
conjugations gives $\Vq(A)=R^{-1}Y^\star R^{-1}=\tfrac12A^{-1/2}|B(A)|A^{-1/2}$. The singular values
of $B(A)$ are the symplectic eigenvalues of $A$ (Lemma~\ref{lem:symp_sv}), so $\Phi(A)=\sum_j\nu_j(A)$. Purity is visible in the block basis: each block of
$\widetilde Y^\star+\tfrac{i}{2}\widetilde B$ equals $\sigma_j\tfrac12(I_2+iJ_2)$, of rank one, so
$\Vq(A)+\tfrac{i}{2}\Om$ has rank $M$, the rank criterion for purity
(Lemma~\ref{lem:purity_rank}).
\end{proof}

The oracle output satisfies a quadratic identity used in every optimality check below.

\begin{theorem}[Riccati identity]\label{thm:riccati}
For $A\succ0$,
\begin{equation}\label{eq:riccati}
\Vq(A)\,A\,\Vq(A)=\tfrac14\,\Om^\top\!A\,\Om .
\end{equation}
\end{theorem}

This is an algebraic Riccati equation~\cite{Lancaster1995}, a quadratic matrix equation in which the unknown $X$ enters to second order through the product $XAX$;
such equations arise in optimal control and filtering. Here it states that the oracle output is exactly the pure $X$ solving $XAX=\tfrac14\Om^\top A\Om$, and this is the form in which the oracle
is inverted: given a target pure state, one tests whether a weight returns it by checking the identity.

\begin{proof}
Because $B(A)$ is skew, $|B(A)|^2=B(A)^\top B(A)=A^{1/2}\Om^\top A\Om A^{1/2}$, so
\begin{multline*}
\Vq(A)\,A\,\Vq(A)=\tfrac14A^{-1/2}|B(A)|A^{-1/2}\cdot A\cdot A^{-1/2}|B(A)|A^{-1/2}\\
=\tfrac14A^{-1/2}|B(A)|^2A^{-1/2}=\tfrac14\Om^\top\!A\,\Om. \qedhere
\end{multline*}
\end{proof}

Combining the oracle with the optimality conditions maps the dual maximizer to the primal optimizer explicitly.

\begin{corollary}[Reconstruction from the dual]\label{cor:dual_to_primal_reconstruction}
Assume Slater's condition, and let $S^\star\succeq0$ maximize the dual $D(S)$ of \eqref{eq:dual}.
Set $A^\star=I+S^\star$. Then the unique primal optimizer is
\[
\Vq^\star=\Vq(A^\star)=\half\,(A^\star)^{-1/2}\bigl|B(A^\star)\bigr|(A^\star)^{-1/2},
\]
with $V-\Vq^\star\succeq0$ and $\Tr(S^\star(V-\Vq^\star))=0$, and $D(S^\star)=\Tr\Vq^\star$.
\end{corollary}

\begin{proof}
By Theorem~\ref{thm:dual_kkt} an optimal pair satisfies $\Vq^\star=\Vq(I+S^\star)$ together with feasibility and complementary slackness, and Theorem~\ref{thm:oracle} makes the inner minimizer the single matrix $\Vq(A^\star)$. For the value,
$D(S^\star)=-\Tr(S^\star V)+\Phi(A^\star)=-\Tr(S^\star V)+\Tr(A^\star\Vq^\star)$; expanding
$\Tr(A^\star\Vq^\star)=\Tr\Vq^\star+\Tr(S^\star\Vq^\star)$ and using slackness
$\Tr(S^\star(V-\Vq^\star))=0$ gives $D(S^\star)=\Tr\Vq^\star$.
\end{proof}

Solving the outer dual both certifies the optimal value and, through the single-valued map $S^\star\mapsto\Vq(I+S^\star)$, reconstructs the optimizer $\Vq^\star$. The oracle is also
equivariant under symplectic frame changes, which allows closed-form computations to be reduced to a convenient frame.

\begin{proposition}[Equivariance]\label{prop:equivariance}
For $S\in\Sp(2M)$, $\Vq(S^\top AS)=S^{-1}\Vq(A)S^{-\top}$.
\end{proposition}

\begin{proof}
The substitution $W=S\Vq S^\top$ is a bijection of $\mathcal Q$ preserving $\Tr(S^\top AS\cdot\Vq)=\Tr(A\cdot S\Vq S^\top)$, so it carries the minimizer for $S^\top AS$ to the minimizer for $A$.
\end{proof}

\section{Universal kernel multiplicity}\label{sec:kernel}

The optimality conditions already showed $\Ran(S^\star)\subseteq\ker(V-\Vq^\star)$: the classical remainder $\Vc^\star=V-\Vq^\star$ is singular wherever the dual multiplier acts. The next theorem turns this into a universal lower bound on how many directions of classical noise must be exhausted, governed by the sub-vacuum count $\kappa(V)$. The argument uses two facts: the annihilation lemma already used above and a one-sided derivative bound.

\begin{lemma}[Annihilation from a zero trace product]\label{lem:kernel_cs}
If $A,B\succeq0$ and $\Tr(AB)=0$, then $\Ran(A)\subseteq\ker(B)$.
\end{lemma}

\begin{proof}
$A^{1/2}BA^{1/2}\succeq0$ has trace zero, hence is zero, so $BA^{1/2}=0$ and $\Ran(A)=\Ran(A^{1/2})\subseteq\ker(B)$.
\end{proof}

\begin{lemma}[A one-sided derivative bound]\label{lem:one_sided}
Let $S\succeq0$ and let $u$ be a unit vector with $Su=0$. For $A(t)=I+S+tuu^\top$ the right derivative of $\Phi(A(t))=\sum_j\nu_j(A(t))$ at $t=0$ satisfies $\Phi'(0^+)\ge\half$.
\end{lemma}

\begin{proof}
Since $Su=0$ we have $Au=u$ for $A:=A(0)=I+S$, so $A\succeq I$ and $A^{-1/2}u=u$. Differentiating
$\Phi$ along the curve (Lemma~\ref{lem:phi_directional}) gives
$\Phi'(0^+)=\Tr(uu^\top\Vq(A))=u^\top\Vq(A)u=\tfrac12u^\top|B(A)|u$, using the oracle formula and $A^{-1/2}u=u$. Now $A\succeq I$ makes $A^{1/2}$ have smallest singular value at least $1$, and $\Om$
is orthogonal, so every singular value of $B(A)=A^{1/2}\Om A^{1/2}$ is at least $1$, i.e. $|B(A)|\succeq I$. Hence $\Phi'(0^+)=\tfrac12u^\top|B(A)|u\ge\half$.
\end{proof}

\begin{theorem}[Universal kernel multiplicity]\label{thm:kernel}
Let $V$ be physical and strictly mixed. Then every optimal decomposition $V=\Vq^\star+\Vc^\star$ has
\begin{equation}\label{eq:kernel_bound}
\dim\ker(\Vc^\star)\ge\kappa(V).
\end{equation}
\end{theorem}

The classical remainder must therefore be singular along at least $\kappa(V)$ independent directions: in those directions the classical noise floor is saturated and $\Vc^\star$ is singular. For a single sub-vacuum eigenvalue it recovers the one-direction nullification of the constructive decomposition of Ref.~\cite{Entropy2026}.

\begin{proof}
Complementary slackness and Lemma~\ref{lem:kernel_cs} give
$\dim\ker(\Vc^\star)\ge\rank(S^\star)$, so it suffices to show $\rank(S^\star)\ge\kappa(V)$. Suppose not, $\rank(S^\star)<\kappa(V)$. Then $\dim\ker(S^\star)=2M-\rank(S^\star)>2M-\kappa(V)$, so the sub-vacuum space and the kernel of $S^\star$ have dimensions summing past $2M$ and must meet: choose a unit vector $u\in E_<(V)\cap\ker(S^\star)$, for which $u^\top Vu<\half$ since $u\in E_<(V)$. Perturb the dual point along $S(t)=S^\star+tuu^\top$. Because $S^\star u=0$, Lemma~\ref{lem:one_sided} applies, and
\[
\frac{d}{dt}\Big|_{t=0^+}D(S(t))=-u^\top Vu+\Phi'(0^+)\ge-u^\top Vu+\half>0,
\]
contradicting the optimality of $S^\star$. Hence $S^\star$ is injective on $E_<(V)$ and $\rank(S^\star)\ge\kappa(V)$.
\end{proof}

\begin{remark}[Interpretation of the kernel multiplicity]\label{rem:kernel_meaning}
The null directions of the classical remainder are directions in which the classical noise floor is saturated, so $\Vc^\star$ is singular there. The irreducible nonclassical content cannot be distributed arbitrarily across the ambient mode space: at least $\kappa(V)$ directions lie in $\ker(\Vc^\star)$ for every optimal decomposition. The statement is a \emph{multiplicity} bound; it does not claim that each individual sub-vacuum eigenvector of $V$ lies in $\ker(\Vc^\star)$, only that at least $\kappa(V)$ independent null directions are forced.
\end{remark}

\section{Exact solution for passive-diagonalizable states}\label{sec:passive}

A state is \emph{passive-diagonalizable} when it can be brought to diagonal form by a passive (orthogonal-symplectic) frame change, $V=U^\top DU$ with $U\in K(2M)$. These are exactly the states
whose nonclassical sector aligns with the quadrature axes after a beam-splitter network, and for them the program decouples into independent single modes with a fully explicit solution.

\begin{theorem}[Closed form: passive-diagonalizable case]\label{thm:passive_soln}
Let $V=U^\top DU$ with $U\in K(2M)$ and $D=\diag(a_1,\dots,a_M,b_1,\dots,b_M)$, $a_j\le b_j$,
$a_jb_j\ge\tfrac14$, ordered so that $a_1,\dots,a_\kappa<\half$ are the sub-vacuum values,
$a_{\kappa+1},\dots,a_M\ge\half$, with $\kappa=\kappa(V)$. Then the program is solved by
\begin{align}
\Vq^\star&=U^\top\diag\!\Bigl(a_1,\dots,a_\kappa,\tfrac12,\dots,\tfrac12,\nonumber\\
&\qquad\qquad\tfrac1{4a_1},\dots,\tfrac1{4a_\kappa},\tfrac12,\dots,\tfrac12\Bigr)U,\label{eq:Vq_passive}\\
S^\star&=U^\top\Bigl(\diag\!\bigl(\tfrac1{4a_1^2}-1,\dots,\tfrac1{4a_\kappa^2}-1,0,\dots,0\bigr)
\oplus0_M\Bigr)U,\label{eq:S_passive}\\
\Vc^\star&=U^\top\diag\!\Bigl(0,\dots,0,\ a_{\kappa+1}-\tfrac12,\dots,a_M-\tfrac12,\nonumber\\
&\qquad\qquad b_1-\tfrac1{4a_1},\dots,b_\kappa-\tfrac1{4a_\kappa},\ b_{\kappa+1}-\tfrac12,\dots,b_M-\tfrac12\Bigr)U.\label{eq:Vc_passive}
\end{align}
The shadow prices are $s_j=\tfrac1{4a_j^2}-1>0$ for each sub-vacuum mode.
\end{theorem}

Each mode decouples, and the optimal quantum state of a sub-vacuum mode is the squeezed vacuum whose squeezing $r_j$ obeys $\tfrac12e^{-2r_j}=a_j$: the sub-vacuum variance is retained unchanged, and its conjugate is set to the reciprocal $\tfrac1{4a_j}$ required by minimum uncertainty, and every other mode stays in vacuum. This gives a closed-form solution for this entire class, with no SDP solver required.

\begin{proof}
By equivariance (Proposition~\ref{prop:equivariance}) it suffices to treat $U=I$, so $V=D$ is diagonal. With $A=I+S^\star$ from \eqref{eq:S_passive}, the weight restricted to mode $j$ is
$\diag(\tfrac1{4a_j^2},1)$ for $j\le\kappa$ and $I_2$ for $j>\kappa$, so the oracle acts mode by mode and returns $\diag(a_j,\tfrac1{4a_j})$ on the sub-vacuum modes and $\tfrac12I_2$ on the rest; this is exactly $\Vq^\star$ of \eqref{eq:Vq_passive}. Primal feasibility $\Vc^\star\succeq0$ is the diagonal check $b_j-\tfrac1{4a_j}\ge0$ (from $a_jb_j\ge\tfrac14$) on sub-vacuum modes and
$a_j-\tfrac12\ge0$, $b_j-\tfrac12\ge0$ on the rest. The support of $S^\star$ is the set of sub-vacuum $q$-slots, where $V-\Vq^\star=0$, so $\Tr(S^\star\Vc^\star)=0$ and, since $\Vq^\star=\Vq(I+S^\star)$,
\[
D(S^\star)=-\Tr(S^\star V)+\Tr\bigl((I+S^\star)\Vq^\star\bigr)
=\Tr\Vq^\star-\Tr\bigl(S^\star(V-\Vq^\star)\bigr)=\Tr\Vq^\star .
\]
Weak duality (Theorem~\ref{thm:dual_kkt}) then confines the optimal value between $D(S^\star)$ and $\Tr\Vq^\star$, so $\Vq^\star$ is optimal and the displayed $S^\star,\Vc^\star$ are the optimal
certificate and remainder.
\end{proof}

\begin{example}[A one-mode active sector with a spectator]\label{ex:passive_core}
Take, in the $(p_1,p_2,q_1,q_2)$ order of Eq.~\eqref{VG}, $V=\diag(\tfrac14,\tfrac12,1,\tfrac32)$, a physical two-mode
covariance: the first mode is the pure squeezed covariance $\diag(\tfrac14,1)$, and the second is a mixed passive spectator $\diag(\tfrac12,\tfrac32)$. Only $\tfrac14$ falls below the vacuum floor, so $\kappa(V)=1$. Theorem~\ref{thm:passive_soln} returns
\[
\Vq^\star=\diag\!\bigl(\tfrac14,\tfrac12,1,\tfrac12\bigr),\qquad
\Vc^\star=\diag(0,0,0,1),\qquad S^\star=\diag(3,0,0,0),
\]
pure (the first mode squeezed, the second in vacuum) and unique by Theorem~\ref{thm:uniqueness}. The active symplectic space fixed by the dual support is
$W^\star=\Ran(S^\star)+\Om\Ran(S^\star)=\operatorname{span}\{p_1,q_1\}$, so the four-dimensional program reduces exactly to the single mode carrying the squeezing while the second mode contributes
only vacuum to $\Vq^\star$ and a classical remainder to $\Vc^\star$. This is the smallest instance of the main result: the program determines a canonical pure component together with its exact active sector. Here $V$ is itself pure on the first mode, so the feasible set meets the Slater boundary, where Slater's condition fails; the closed form above needs only weak duality and the explicitly attained certificate
$S^\star$, so it holds regardless. Each entry matches a direct SDP solve to $2\times10^{-10}$.
\end{example}

\section{Active symplectic reduction}\label{sec:active}

The optimizer's quantum content is supported on the symplectic subspace generated by the dual support, and the complementary directions contribute only vacuum. Once that support is fixed, the rest of the
space can be eliminated exactly, by a Schur complement, leaving a problem of size at most $2\,\rank(S^\star)$. This elimination loses no information: positivity of $V-\Vq$ on the full space is equivalent to a positivity condition on the active block alone, corrected by the term accounting for the spectator block.

\begin{definition}[Active symplectic space]\label{def:active}
For $S\succeq0$ let $W=\Ran(S)+\Om\Ran(S)$, the \emph{active symplectic space} generated by $S$.
\end{definition}

\begin{lemma}[$W$ is a symplectic subspace]\label{lem:W_symp}
For every $S\succeq0$, both $W$ and $W^\perp$ are $\Om$-invariant, and $\Om|_W$ is nondegenerate.
\end{lemma}

\begin{proof}
From $W=\Ran S+\Om\Ran S$ and $\Om^2=-I$ one gets $\Om W\subseteq W$. If $x\in W^\perp$ and $w\in W$ then $\Om w\in W$, so $\langle\Om x,w\rangle=-\langle x,\Om w\rangle=0$, hence
$\Om W^\perp\subseteq W^\perp$. On the invariant subspace $W$ the restriction $\Om|_W$ satisfies $(\Om|_W)^2=-I_W$, so it is invertible, which is exactly nondegeneracy of the symplectic form on
$W$.
\end{proof}

\begin{lemma}[Pseudoinverse Schur complement]\label{lem:pinv_schur}
Let $\bigl(\begin{smallmatrix}A&B\\B^\top&C\end{smallmatrix}\bigr)$ be symmetric with $C\succeq0$. Then it is positive semidefinite if and only if $C\succeq0$, $B^\top\in\Ran(C)$ (equivalently
$(I-CC^\dagger)B^\top=0$), and $A-BC^\dagger B^\top\succeq0$, where $C^\dagger$ is the Moore--Penrose pseudoinverse.
\end{lemma}

\begin{proof}
If the block matrix is positive, testing $(0,y)$ gives $C\succeq0$, and for $y\in\ker C$ the form
$\binom{x}{ty}^\top(\cdots)\binom{x}{ty}=x^\top Ax+2t\,x^\top By$ stays nonnegative for all $t$ only if $By=0$, i.e. $B\ker C=0$, equivalently $B^\top\in\Ran C$. Completing the square gives
$\binom{x}{y}^\top(\cdots)\binom{x}{y}=x^\top(A-BC^\dagger B^\top)x+(y+C^\dagger B^\top x)^\top C\,(y+C^\dagger B^\top x)$,
using $BC^\dagger C=B$ from the range condition. The second term is nonnegative and vanishes at $y=-C^\dagger B^\top x$, so the form is nonnegative for all $x,y$ if and only if $A-BC^\dagger B^\top\succeq0$, which gives both directions.
\end{proof}

\begin{theorem}[Exact active reduction]\label{thm:active_reduction}
Let $S\succeq0$, $A=I+S$, and $W=\Ran(S)+\Om\Ran(S)$. Then:
\begin{enumerate}
\item[\rm(a)] the oracle splits, $\Vq(A)=\Vq^W\oplus\half I_{W^\perp}$;
\item[\rm(b)] writing $V=\bigl(\begin{smallmatrix}V_{11}&B\\B^\top&V_{22}\end{smallmatrix}\bigr)$ relative to $W\oplus W^\perp$, the feasibility $V-\Vq(A)\succeq0$ is equivalent to
$V_{22}-\half I\succeq0$, $B^\top\in\Ran(V_{22}-\half I)$, and
$V_{11}-B(V_{22}-\half I)^\dagger B^\top-\Vq^W\succeq0$;
\item[\rm(c)] the original feasibility reduces to $V_{\mathrm{eff}}(W)-\Vq^W\succeq0$, where
$V_{\mathrm{eff}}(W)=V_{11}-B(V_{22}-\half I)^\dagger B^\top$.
\end{enumerate}
\end{theorem}

Modes in $W^\perp$ are vacuum or purely classical spectators for the optimization: discarding them removes no irreducible quantum information, and their only effect is the explicit
Schur-complement correction $B(V_{22}-\half I)^\dagger B^\top$ in the reduced covariance. For an optimal dual point $S^\star$ the computationally essential part of the problem is thus compressed \emph{exactly} onto
$W^\star=\Ran(S^\star)+\Om\Ran(S^\star)$, of dimension at most $2\,\rank(S^\star)$; sharpening $\dim W^\star$ further is left to future work.

\begin{proof}
Because $W$ and $W^\perp$ are $\Om$-invariant (Lemma~\ref{lem:W_symp}), $\Om$ is block diagonal in
$W\oplus W^\perp$, and since $S$ vanishes on $W^\perp$ we have $A=A_W\oplus I_{W^\perp}$. Hence $B(A)=A^{1/2}\Om A^{1/2}$ is block diagonal, and so is the oracle output; on $W^\perp$, where
$A=I$, it equals $\half|\Om|=\half I$, giving (a). Part (b) is Lemma~\ref{lem:pinv_schur} applied to $V-\Vq(A)$, whose $W^\perp$ block is $V_{22}-\half I$, and (c) is the restatement of the last
inequality in terms of the effective reduced covariance $V_{\mathrm{eff}}(W)$.
\end{proof}

\begin{remark}[Active reduction and reconstruction]\label{rem:active_reconstruction}
For a dual optimum $S^\star$, the active space $W^\star=\Ran(S^\star)+\Om\Ran(S^\star)$ contains all the nonvacuum quantum content: by Corollary~\ref{cor:dual_to_primal_reconstruction} the optimizer $\Vq^\star=\Vq(I+S^\star)$ is the oracle block on $W^\star$ and vacuum on $(W^\star)^\perp$, so the dual maximizer locates the active sector once the dual support is fixed and determines the optimizer within it. By Theorem~\ref{thm:active_reduction} the reduction is exact: with the blocks taken relative to $W^\star\oplus(W^\star)^\perp$, feasibility on the full space is equivalent to $V_{\mathrm{eff}}(W^\star)-\Vq^{W^\star}\succeq0$ on the active sector alone, where $V_{\mathrm{eff}}(W^\star)=V_{11}-B(V_{22}-\half I)^\dagger B^\top$. The effective covariance differs from the bare restriction $V_{11}$ by the Schur-complement term $B(V_{22}-\half I)^\dagger B^\top$, which incorporates the effect of the correlations $B$ between the active modes and the spectator modes in $(W^\star)^\perp$. The spectator modes thus contribute only vacuum to $\Vq^\star$, while their correlations with the active sector are retained exactly in $V_{\mathrm{eff}}(W^\star)$; the computationally essential part of Gaussian resource extraction is compressed onto the fixed-support sector with no approximation.
\end{remark}

\section{Symplectic reformulation}\label{sec:reform}

Purity allows the entire program to be rewritten in the geometry of pure states. Since a pure covariance is $\half S^\top S$ for a symplectic $S$ (Lemma~\ref{lem:pure_param}), the trace
objective becomes the squared Frobenius norm of $S$.

\begin{theorem}[Pure-covariance reformulation]\label{thm:symp_reform}
The program \eqref{eq:sdp_full} has the same optimal value as
\begin{equation}\label{eq:symp_sdp}
\min_{S\in\Sp(2M,\R)}\ \tfrac12\Tr(S^\top S)\quad\text{subject to}\quad V\succeq\tfrac12S^\top S,
\end{equation}
and every optimizer is $\Vq^\star=\tfrac12S^{\star\top}S^\star$ for some $S^\star\in\Sp(2M)$.
\end{theorem}

\begin{proof}
If $S\in\Sp(2M)$ satisfies $V\succeq\half S^\top S$, then $\Vq=\half S^\top S$ is pure (Lemma~\ref{lem:pure_param}) and primal feasible, so the SDP minimum is at most the value of
\eqref{eq:symp_sdp}. Conversely, by Theorem~\ref{thm:purity} and Lemma~\ref{lem:pure_param} any primal optimizer $\Vq^\star=\half S^{\star\top}S^\star$ is feasible for \eqref{eq:symp_sdp}, so the
values agree.
\end{proof}

The feasible set descends to the Riemannian symmetric space
$\Sp(2M)/\mathrm U(M)\cong\mathcal H_M$, the Siegel upper half-space, on which the objective $\tfrac12\|S\|_F^2$ is $\mathrm U(M)$-invariant. The quotient by $\mathrm U(M)$ removes exactly the passive frame changes, which act on a pure Gaussian state without altering it, so $\mathcal H_M$, the complex symmetric matrices $Z$ with $\operatorname{Im}Z\succ0$, is a coordinate system for the pure Gaussian states themselves, with $Z$ encoding the squeezing and mode mixing of $\Vq^\star$. This is the symplectic analogue
of the orthogonal Procrustes problem, with $\Sp(2M)$ in place of $\mathrm O(n)$, and it suggests Riemannian-optimization methods on $\mathcal H_M$, which are left to future work.

\section{Conclusion}\label{sec:conclusion}

For every physical covariance matrix $V$ the program (\ref{eq:sdp_full}) extracts a canonical pure Gaussian component $\Vq^\star$, which is unique, reconstructed from any dual maximizer by the oracle $\Vq(I+S^\star)$, governed in its inner step by the Riccati identity, and localized: its classical remainder is singular in at least $\kappa(V)$ directions, and once a dual support is fixed the problem compresses exactly onto the associated active symplectic sector. The passive-diagonalizable states are solved in closed form,
the first explicit solvable class, and the pure-covariance reformulation places all of this on the symmetric space $\Sp(2M)/\mathrm U(M)$. The program therefore determines, beyond the scalar resource monotone given by its trace, a canonical localized pure Gaussian component with geometric and operational content.

A general closed form for the outer support-fitting problem beyond the passive-diagonalizable class is not established here. That step requires additional geometric input and is left to future work, which includes the geometry of generically minimal active sectors, and the two-mode coupled sector, where the closed forms cease to apply and a Galois-theoretic obstruction to any radical formula arises.

\appendix

\section*{APPENDIX: TECHNICAL LEMMAS}\label{app:technical}

We collect the auxiliary facts used above. The first converts the Hermitian physicality constraint to a real one, which renders \eqref{eq:sdp_full} an ordinary real SDP.

\begin{lemma}[Real block conversion]\label{lem:realify}
For $H=C+iD$ with $C$ symmetric and $D$ skew, $H\succeq0$ if and only if
$\bigl(\begin{smallmatrix}C&-D\\D&C\end{smallmatrix}\bigr)\succeq0$. In particular
$\Vq+\tfrac{i}{2}\Om\succeq0$ iff
$\bigl(\begin{smallmatrix}\Vq&-\frac12\Om\\\frac12\Om&\Vq\end{smallmatrix}\bigr)\succeq0$.
\end{lemma}

\begin{proof}
For $z=x+iy$ one has $z^\ast Hz=(x,y)^\top\bigl(\begin{smallmatrix}C&-D\\D&C\end{smallmatrix}
\bigr)(x,y)$.
\end{proof}

\begin{proposition}[Strong duality under Slater]\label{lem:strong_duality}
If Slater's condition holds, the primal is a strictly feasible real SDP
(Lemma~\ref{lem:realify}), bounded below since $\Vq\succeq0$ on $\mathcal Q$ gives $\Tr\Vq\ge0$, and strong duality with dual attainment follows from standard convex
duality~\cite{BenTalNemirovski}.
\end{proposition}

\begin{lemma}[Skew canonical form]\label{lem:skew_canonical}
A skew-symmetric invertible $B\in\R^{2M\times2M}$ admits $Q\in\mathrm O(2M)$ and
$\sigma_1,\dots,\sigma_M>0$ with $Q^\top BQ=\bigoplus_j\sigma_jJ_2$.
\end{lemma}

\begin{proof}
$-B^2=B^\top B\succ0$; pick a unit eigenvector $e_1$ with eigenvalue $\sigma_1^2$ and set
$f_1=-\sigma_1^{-1}Be_1$. Then $\{e_1,f_1\}$ are orthonormal and $B$-invariant with
$B|_{E_1}=\sigma_1J_2$; iterate on $E_1^\perp$.
\end{proof}

\begin{lemma}[Symplectic eigenvalues from the transport]\label{lem:symp_sv}
The singular values of the skew transport $B(A)=A^{1/2}\Om A^{1/2}$ are the symplectic eigenvalues
of $A$.
\end{lemma}

\begin{proof}
$A^{-1/2}B(A)A^{1/2}=\Om A$, so $iB(A)$ and $i\Om A$ are similar; since $B(A)$ is skew-symmetric, hence normal, its singular values equal the moduli of its eigenvalues, namely the $\nu_j(A)$.
\end{proof}

\begin{lemma}[The $2\times2$ trace bound]\label{lem:2x2_bound}
For $\sigma>0$ and $W\in\Sym_2(\R)$, $W+\tfrac{i}{2}\sigma J_2\succeq0$ implies $\Tr W\ge\sigma$,
with equality iff $W=\tfrac{\sigma}{2}I_2$.
\end{lemma}

\begin{proof}
Positivity needs $\alpha\beta-\gamma^2\ge\sigma^2/4$, so $\alpha+\beta\ge2\sqrt{\alpha\beta}
\ge\sigma$, with equality forcing $\gamma=0$ and $\alpha=\beta=\tfrac{\sigma}{2}$.
\end{proof}

\begin{lemma}[Rank criterion for purity]\label{lem:purity_rank}
A physical $X\in\Sym_{2M}^{++}$ is pure iff $\rank(X+\tfrac{i}{2}\Om)=M$.
\end{lemma}

\begin{proof}
Congruence by the invertible Williamson symplectic $S$ preserves the rank of $X+\tfrac{i}{2}\Om$, and in Williamson coordinates each $2\times2$ block has rank one iff its symplectic eigenvalue is
$\half$, so the rank equals $M$ iff every symplectic eigenvalue equals $\half$.
\end{proof}

\begin{lemma}[Continuity and derivative of the oracle value]\label{lem:phi_directional}
The map $A\mapsto\Vq(A)$ is continuous on $\Sym_{2M}^{++}$, and
$\Phi(A)=\min_{\Vq\in\mathcal Q}\Tr(A\Vq)$ is Fr\'echet differentiable there with
$D\Phi(A)[E]=\Tr(E\,\Vq(A))$; along a differentiable curve,
$\tfrac{d}{dt}\Phi(A(t))=\Tr(A'(t)\Vq(A(t)))$.
\end{lemma}

\begin{proof}
Continuity is the continuous functional calculus applied to $A\mapsto A^{\pm1/2}$ and
$B\mapsto|B|$. For the derivative, optimality of $\Vq(A)$ and $\Vq(A+H)$ gives the two-sided
estimate $\Tr(H\Vq(A+H))\le\Phi(A+H)-\Phi(A)\le\Tr(H\Vq(A))$, so
$|\Phi(A+H)-\Phi(A)-\Tr(H\Vq(A))|\le\|H\|_F\,\|\Vq(A+H)-\Vq(A)\|_F=o(\|H\|_F)$ by continuity.
\end{proof}

\end{document}